\documentclass[a4paper]{article}
\usepackage[top=25truemm, bottom=25truemm, left=20truemm, right=20truemm]{geometry}

\usepackage{graphicx}%
\usepackage{multirow}%
\usepackage{amsmath,amssymb,amsfonts}%
\usepackage{amsthm}%
\usepackage{mathrsfs}%
\usepackage[title]{appendix}%
\usepackage{xcolor}%
\usepackage{textcomp}%
\usepackage{manyfoot}%
\usepackage{booktabs}%
\usepackage{algorithm}%
\usepackage{algorithmicx}%
\usepackage{algpseudocode}%
\usepackage{listings}%

\usepackage[legacycolonsymbols]{mathtools}
\usepackage{siunitx}

\usepackage{natbib} 

\usepackage[colorlinks=true,citecolor=blue,linkcolor=blue,breaklinks=true]{hyperref} 

\newtheorem{theorem}{Theorem}
\newtheorem{remark}{Remark}%

\newtheorem{definition}{Definition}%
\newtheorem{lemma}{Lemma}

\title{\textbf {Accurate and Efficient Approximation\\ of the Null Distribution of Rao's Spacing Test}}
\date{}
\author{
\textbf {Yoshiki Kinoshita{$^1$}, Aya Shinozaki{$^2$} and Toshinari Kamakura{$^3$}}\\
\normalsize{{${}^1$yoshikik.92g@g.chuo-u.ac.jp}; {${}^2$shinozaki.223@g.chuo-u.ac.jp}; {${}^3$ktoshinari001f@g.chuo-u.ac.jp}} \\
\normalsize{{{${}^{1,2,3}$}Chuo University, 1-13-27 Kasuga, Bunkyo-ku, Tokyo 112-8551 Japan}}
}

\begin{document}

\maketitle

\abstract{

Rao's spacing test is a widely used nonparametric method for assessing uniformity on the circle. 
However, its broader applicability in practical settings has been limited because the null distribution is not easily calculated. 
As a result, practitioners have traditionally depend on pre-tabulated critical values computed for a limited set of sample sizes, which restricts the flexibility and generality of the method. 
In this paper, we address this limitation by recursively computing higher-order moments of the Rao's spacing test statistic and employing the Gram-Charlier expansion to derive an accurate approximation to its null distribution. 
This approach allows for the efficient and direct computation of P-values for arbitrary sample sizes, thereby eliminating the dependency on existing critical value tables. 
Moreover, we confirm that our method remains accurate and effective even for large sample sizes that are not represented in current tables, thus overcoming a significant practical limitation. 
Comparative evaluations with published critical values and saddlepoint approximations demonstrate that our method achieves a high degree of accuracy across a wide range of sample sizes. 
These findings greatly improve the practicality and usability of Rao's spacing test in both theoretical investigations and applied statistical analyses. 
}

\section{Introduction}\label{sec1}

Testing for uniformity on the circle is a fundamental problem in the analysis of directional data. 
For unimodal distributions, the von Mises distribution is widely used in the literature \cite{mardia2009directional, jammalamadaka2001topics}. 
The Rayleigh test for uniformity , which is derived from the likelihood ratio test is known to be powerful under the von Mises distribution. 
However, when the data are multimodal, nonparametric methods such as Rao's spacing test have been shown to be more effective than parametric alternatives like the Rayleigh test.

Rao's spacing test has seen limited practical application due to the difficulty of numerically calculating the distribution of its test statistic under the null hypothesis.
As a result, practitioners must consult tabulated critical values that are only available for specific sample sizes, as provided by \cite{russell1995expanded}. 
This constraint severely limits the test's flexibility.

In this paper, we calculate higher-order moments of the Rao's spacing test statistic recursively and apply the Gram-Charlier expansion to approximate its null distribution. 
This allows us to compute P-values directly, thus enabling the application of Rao's spacing test to any sample size. 
We validated our method by comparing the derived approximations with the existing tables from \cite{russell1995expanded} and those from \cite{gatto1999conditional}. 
The comparisons show that our proposed approach achieves accuracy comparable to or better than theirs. 
In addition, the existing tables are limited to sample sizes up to $n=1000$, and for larger samples, the critical value for $n=1000$ is commonly used as an approximation. 
However, the effect of this substitution has not been thoroughly investigated. 
In this study, we examine the impact of this approximation by comparing our proposed method with the existing approach for larger samples, specifically when $n=10000$. 

\section{Moments of Rao's Spacing Test}\label{sec2}

The Rao's spacing test statistic plays a central role in testing uniformity on the circle. Although its probability density function (PDF) is known, both the PDF and the cumulative distribution function (CDF) are difficult to compute in practice. 
Consequently, critical values have been tabulated for specified sample sizes and angle settings by \cite{russell1995expanded}, and made available through implementations in R and MATLAB.
Their method, based on numerical integration, is efficient for small sample sizes. However, as the sample size increases, the recursion terms become computationally expensive, making real-time evaluation impractical. Thus, users are typically limited to the discrete settings included in the precomputed tables.
To address this limitation, we derive higher-order moments of the Rao's spacing test statistic under the null hypothesis and construct an approximate distribution using the Gram-Charlier expansion. This approach enables testing at arbitrary sample sizes without dependence on tabulated values. 

This section begins with an introduction to Rao's spacing test statistic. We then provide a general expression for its higher-order moments. Lastly, we derive the Gram-Charlier expansion of the CDF.

\subsection{Rao's Spacing Test Statistic}\label{subsec1}

Let $\alpha_1, \dots, \alpha_n$ be observations on the circle, ordered such that:
\begin{eqnarray*}
0\le \alpha_{(1)}\le \cdots \le\alpha_{(n)}<2\pi. 
\end{eqnarray*}

The Rao's spacing test statistic is defined as:
\begin{eqnarray}\label{eq1}
U_n=\frac{1}{2}\sum_{i=1}^n\left|T_i-\frac{2\pi}{n} \right|, 
\end{eqnarray}
where 
$T_i=\alpha_{(i+1)}-\alpha_{(i)} (i=1, \dots, n-1)$ and 
$T_n=\alpha_{(1)}-\alpha_{(n)}+2\pi$. 

While the exact PDF of $U_n$ is known, its evaluation is computationally intensive. 
The density function of $U_n$ is known to be: 
\begin{eqnarray*}
f_n(u)=(n-1)!\sum_{j=1}^{n-1}\binom{n}{j}\left(\frac{u}{2\pi}\right)^{n-j-1}
\frac{\phi_j(nu)}{(n-j-1)!n^{j-1}} \quad \left(0\le u\le 2\pi\left(1-\frac{1}{n}\right)\right), 
\end{eqnarray*}
where
\begin{eqnarray*}
\phi_j(x)=\frac{1}{2\pi(j-1)!}\sum_{k=0}^{j-1}(-1)^k\binom{j}{k}\left[\max\left(0,  \frac{x}{2\pi}-k\right) \right]^{j-1}. 
\end{eqnarray*}
Here, $\phi_j(x)$ denotes the Irwin-Hall distribution, which is the distribution of the sum of independent uniform random variables on $[0, 2\pi]$.

\subsection{Derivation of Higher-Order Moments}\label{subsubsec1}

\cite{Rao1969} provided a table of critical points for the Rao's spacing test statistic by modifying the values originally reported in \cite{sherman1950random}. 
However, this table is limited to sample sizes ranging from $n=2$ to $n=20$, with only three significance levels ($\alpha = 0.01$, $0.05$, and $0.10$). 
This limitation stems from the fact that the exact density function of the test statistic involves complex expressions, making the calculation of critical values analytically challenging.

To address these limitations, \cite{russell1995expanded} employed the trapezoidal rule to numerically compute critical values for sample sizes ranging from $n=4$ to $n=1000$. 
They proposed a recursive method to evaluate part of the trapezoidal integration more efficiently. 
Nevertheless, the computation of the critical points remained time-consuming, making it impractical to calculate them on the fly and necessitating the use of precomputed tables.

In response, we derived a recursive method for computing the higher-order moments of Rao's spacing test statistic and applied the Gram-Charlier expansion, an asymptotic series, to approximate its distribution.
This approach enables the rapid computation of P-values, which was previously infeasible due to the complexity of the distribution. 
Our method thus provides a significant computational advantage, especially for large-scale applications.

In this subsection, we derive the following theorem to compute the $r$th moment of Rao's spacing test statistic recursively. 

\begin{theorem}\label{col}
\begin{align}
E(U_n^r)=\frac{(2\pi)^r}{n^{n+r-1}(n)_{\overline{r}}}\sum_{j=1}^ra_j^{(r)}(n)_{\underline{j}}(n-j)^{n+r-1}\notag
\end{align}
where $a_j^{(r)}$ satisfies: 
\begin{align*}
	a_{r}^{(r)}&=1, \\
	a^{(r+1)}_1&=(r+1)a_1^{(r)},\\
	a^{(r+1)}_j&=(r+j)a_j^{(r)}+a_{j-1}^{(r)}\qquad(j=2,\ldots,r), 
\end{align*}
and for all $n\in\mathbb{N}_{\geq 2}$ and $r\in\mathbb{N}_{\geq 1}$. 
\end{theorem}

\begin{remark}
Theorem~\ref{col} is expressed as a linear combination of $(n)_{\underline{j}}(n-j)^{n+r-1}, (j=1, \dots, r)$, and the coefficients are as follows: 
\[
\frac{(2\pi)^r}{n^{n+r-1}(n)_{\overline{r}}}a_j^{(r)}(n)_{\underline{j}} \quad (j=1, \dots, r). 
\]
\end{remark}

The proof of Theorem~\ref{col} is provided in the appendix~\ref{secB}.

The first and second moments can be calculated by applying the general formula from Theorem~\ref{col}, using the coefficients listed in Table~\ref{table_a}, as follows. 
\begin{eqnarray*}
E(U_n)&=&2\pi\left(1-\frac{1}{n}\right)^n, \\
E(U_n^2)&=&\frac{(2\pi)^2}{n^{n+1}(n+1)}\left(2(n-1)^{n+1}+(n-1)(n-2)^{n+1}\right). 
\end{eqnarray*}

\subsection{Gram-Charlier Expansion of the CDF}
Using raw moments $\mu^r=E(U_n^r)$, we compute cumulants and construct the Gram-Charlier expansion:
\begin{eqnarray*}
f_n(u)&=&\sum_{j=3}^\infty B_j(0,0,\kappa'_{n,3},\ldots,\kappa'_{n,j})\frac{H_j(u)}{j!}\phi(u), \\
F_n(u)&=&\Phi(u)-\phi(u)\sum_{j=3}^\infty B_j(0,0,\kappa'_{n,3},\ldots,\kappa'_{n, j})\frac{H_{j-1}(u)}{j!}, 
\end{eqnarray*}
where $\phi$ and $\Phi$ are the standard normal PDF and CDF, respectively and $B_j$ is the $j$th complete exponential Bell polynomials. 
Furthermore, $H_j (j=0,1,\ldots)$ represents the Hermite polynomials, as defined by 
\begin{align*}
	H_0(x)&=1, \\
	H_{(j+1)}(x)&=xH_j(x)-\frac{d}{dx}H_j(x)\qquad(j=0,1,\ldots). 
\end{align*}

This expansion allows for the accurate computation of P-values without dependence on tabulated critical values. 
When $n$ is small (say $6$), this approximation may be inaccurate; in such cases, an exact result can be obtained by directly integrating the PDF of $f_n(u)$. 

\section{Comparison with Existing Methods}\label{sec3}

In this section, we assess the performance of the proposed Gram-Charlier approximation by comparing it with existing methods. 
Specifically, we compare our method with the tabulated values provided by \cite{russell1995expanded} and the saddlepoint approximations proposed by \cite{gatto1999conditional}.

\begin{table}[!htbp]
\centering
\caption{Comparison of the accuracy between the Gram-Charlier approximation and the saddlepoint approximation. $(n=10)$}
\label{table_saddle}
\begin{tabular}{r||r|rrr|r}
\multicolumn{1}{c}{\text{Angle}}&\multicolumn{1}{c}{\text{Exact}}&\multicolumn{3}{c}{\text{Saddlepoint Approximations}}&{\text{Our method}}\\
$t$&$\Pr_E(T<{t})$&$\Pr_{LR}(T<{t})$&$\Pr_{BN}(T<{t})$&$\Pr_{ID}(T<{t})$&$\Pr_{G_{10}}(T<{t})$\\\hline
\ang{50}&.001&.001&.001&0&.001\\
\ang{60}&.004&{.004}&{.004}&{.004}&{.004}\\
\ang{70}&.015&{.015}&{.015}&{.015}&{.015}\\
\ang{80}&.042&.041&.041&.043&{.042}\\
\ang{90}&.093&{.093}&.092&.096&{.093}\\
\ang{100}&.178&.176&.175&.182&{.178}\\
\ang{110}&.294&.292&.291&.300&{.294}\\
\ang{120}&.433&.430&.429&.439&{.433}\\
\ang{130}&.577&.580&.578&.583&{.577}\\
\ang{140}&.708&.706&.705&.713&{.708}\\
\ang{150}&.815&.813&.812&.819&{.815}\\
\ang{160}&.892&.891&.890&.895&{.892}\\
\ang{170}&.943&.942&.941&.945&{.943}\\
\ang{180}&.972&{.972}&.971&.973&{.972}\\
\ang{190}&.988&.987&.987&.989&{.988}\\
\ang{200}&.995&{.995}&{.995}&.996&{.995}\\
\ang{210}&.998&{.998}&{.998}&.999&{.998}\\
\ang{220}&.999&{.999}&{.999}&1.000&{.999}\\\hline
\end{tabular}
\end{table}

Table~\ref{table_saddle} presents the cumulative distribution function values for several values of the Rao's spacing test statistic when the sample size is fixed at $n = 10$. 
The column $\text{Pr}_{E}(T< t)$ shows the exact values from \cite{russell1995expanded}, while $\text{Pr}_{LR}(T<t)$, $\text{Pr}_{BN}(T<t)$, and $\text{Pr}_{ID}(T<t)$ represent saddlepoint approximations from \cite{gatto1999conditional}. 
The column $\text{Pr}_{G_{10}}(T < t)$ gives the Gram-Charlier approximations using moments up to the $10$th order. 
The results indicate that our Gram-Charlier expansion using up to the $10$th moment matches the exact probabilities across the entire range. 
Moreover, our method outperforms the saddlepoint approximations in terms of accuracy in all cases. 

\section{Examples}\label{sec4}
We conduct a uniformity test for the example in \cite{russell1995expanded}  based on the method we have derived. 
The dataset consists of hypothetical data for an epileptic patient, recording the times of seizure on sets. 
The angular values are as follows: 
\[
\ang{5}, \ang{10}, \ang{10}, \ang{12}, \ang{17}, \ang{85}, \ang{90}, \ang{99}, \ang{100}, \ang{110}, \ang{153}, \ang{233}, \ang{235}, \ang{296}, \ang{331}.
\] 
In our method, when the 10th moment ($r=10$) is used, the example results in a sample size of $n=15$ and $U_{15} = 3.089233(=\ang{177})$, as shown in Equation \eqref{eq1}. 
The corresponding P-value is 0.0174, leading to the rejection of the null hypothesis of uniformity. 

We also conduct a uniformity test for Example 6.9 from \cite{mardia2009directional}, based on the method we have derived. 
This example involves a pigeon-homing experiment in which 13 birds were released, and their vanishing angles were recorded as follows:
\[
\ang{20}, \ang{135}, \ang{145}, \ang{165}, \ang{170}, \ang{200}, \ang{300}, \ang{325}, \ang{335}, \ang{350}, \ang{350}, \ang{350}, \ang{355}.
\]
In our method, when the 10th moment ($r=10$) is used, the example results in a sample size of $n=13$ and $U_{13} = 2.826091(=\ang{161.9231})$, as shown in Equation~\eqref{eq1}. The corresponding P-value is 0.0786, which does not lead to the rejection of the null hypothesis of uniformity. Our method allows for the calculation of the P-value, enabling the uniformity test to be conducted without the need for statistical tables. In contrast, \cite{mardia2009directional} refer to the tables provided by \cite{Rao1969} to obtain the results of the uniformity test. 
They concluded that the hypothesis of uniformity was not rejected at the 5\% level of significance, since the test statistic $U_{13}=\ang{162}$ lies between the the 5\% critical value ($\ang{167.76}$) and the 10\% critical value ($\ang{158.40}$) given in Rao (1969)'s table. 

We also conduct a uniformity test for an example from \cite{Kitamura1989circular}, based on the method we have derived. 
This example involves Kamiokande neutrino data, and the dataset is as follows:
\begin{eqnarray*}
&&\ang{30}, \ang{36}, \ang{60}, \ang{64}, \ang{76}, \ang{98}, \ang{136}, \ang{140}, \ang{182}, \ang{216}, \ang{244}, \ang{270} (n=12), \\
&&\ang{30}, \ang{36}, \ang{60}, \ang{64}, \ang{76}, \ang{98}, \ang{140}, \ang{182}, \ang{216}, \ang{244}, \ang{270} (n=11). 
\end{eqnarray*}
In our method, when the 1st to 10th moments ($r=10$) are used, the example results in a sample size of $n=12$ and $U_{12} = 1.989675(=\ang{114})$, and a sample size of $n=11$ and $U_{11} = 1.869089 (=\ang{107.0909})$, as shown in Equation \eqref{eq1}. 
The corresponding P-values are $0.685 (n=12)$ and $0.762 (n=11)$, which does not lead to the rejection of the null hypothesis of uniformity. 
Our method allows for the calculation of the P-value, enabling the uniformity test to be conducted without the need for statistical tables. 
In contrast, \cite{Kitamura1989circular} refers to the tables provided by \cite{Rao1969} to obtain the results of the uniformity test. 
He notes that the test statistic for the Kamiokande data are both below the critical value of $\alpha=0.10$, and therefore, the significance of deviations from randomness cannot be revealed. 
He also claimed that there are no P-values corresponding to the statistic in Rao's  table. 
Therefore, our method, which allows for the calculation of P-values, is considered useful for determining whether the null hypothesis should be rejected in the Kamiokande data. 
In case of the Kamiokande data, the derived P-values are both very large, leading us to conclude that the hypothesis should not be rejected.

\section{Simulation Study}
In the circular package in R, the Rao's uniformity test refers to the \cite{russell1995expanded} tables, which only provide critical values up to a sample size of $n = 1000$. 
When performing the test for sample sizes larger than $1000$, the critical value for $n = 1000$ is used as a substitute. 
To evaluate the impact of this, we carry out a simulation with a sample size of $n = 10000$ and compare our proposed method with the circular package in R. 

For $n = 1000$, the rejection critical values from the \cite{russell1995expanded} tables coincide with those of our proposed method. 
For $n = 10000$, $d=10$, the critical values are listed in Table~\ref{table_rejection_critical_values}. 
\begin{table}[!htbp]
\centering
\caption{Critical values. }\label{table_rejection_critical_values}
\begin{tabular}{l|rrrr}
Significant levels &0.001&0.01&0.05&0.10\\\hline
\cite{russell1995expanded} ($n=1000$)&140.99&138.84&136.94&135.92\\
our proposed method ($n=1000$)&140.99&138.84&136.94&135.92\\
\cite{russell1995expanded} ($n=2000$)&NA&NA&NA&NA\\
our proposed method ($n=2000$)&138.49&136.97&135.63&134.91\\
\cite{russell1995expanded} ($n=10000$)&NA&NA&NA&NA\\
our proposed method ($n=10000$)&135.14&134.47&133.87&133.55\\
\end{tabular}
\end{table}

For $n = 10000$, the critical value is smaller compared to $n = 1000$,  suggesting that as the sample size increases, the rejection region becomes larger. 
If, despite having a sample size of $10000$, the table for $n = 1000$ is used as a substitute, it may result in failing to reject the null hypothesis in situations where it should correctly be rejected. 

To compare the rejection rates, we generate a uniform distribution on the unit circle with $n = 10000$ and perform the test at a significance level of $\alpha = 0.05$. 
This procedure is repeated $10000$ times, and the number of rejections is recorded in Table~\ref{table_alpha_005}. 

\begin{table}[!htbp]
\centering
\caption{Number of acceptances and rejections. }\label{table_alpha_005}
\begin{tabular}{l|rrr}
&\mbox{accept}&\mbox{reject}&repetations \\\hline
\mbox{circular package} &10000&0&10000\\
\mbox{our proposed method} &9482&518&10000\\
\end{tabular}
\end{table}

The results show that the rejection rate of our proposed method generally aligns with the set significance level ($\alpha = 0.05$), whereas the circular package shows no rejections at all. 
Since this test is performed on uniform distribution data, the absence of rejections is not an issue. 
However, this suggests that the detection power for distributions that should be rejected is expected to be lower when using the circular package. 

Here, we compare the detection power for the von Mises distribution with a sample size of $10000$ and parameters (mean parameter $\theta=0$, concentration parameter $\kappa=0.3$) in Table~\ref{table_power}. 

\begin{table}[!htbp]
\centering
\caption{Number of acceptances and rejections. }\label{table_power}
\begin{tabular}{l|rrr}
&\mbox{accept}&\mbox{reject}&repetations\\\hline
\mbox{circular package} &9624&376&10000\\
\mbox{our proposed method} &514&9486&10000\\
\end{tabular}
\end{table}

The results show that the rejection rate of our proposed method is higher, while the rejection rate of the circular package is lower, indicating that our proposed method has higher detection power. 

\section{Discussion and Conclusion}\label{sec5}






Our study has made it possible to derive the P-values for Rao's test statistic in the uniformity test on a circle, that was previously not feasible. 
Up to now, the only available approach for using this test was by referring to tabulated values($n\le 1000$), which appeared in the R circular package based on \cite{russell1995expanded}. 
However, with the P-values we have derived, it is no longer necessary to consult such tables. 
Therefore, our new approach is expected to lead to new developments across various fields of applications, compared to the traditional method using existing tables. 

Furthermore, the extension of Rao's spacing test to the sphere has not yet been accomplished. 
Problems involving angles on the sphere, such as the directions of palaeomagnetism in rock, the directions from the earth to stars, and the directions of optical axes in quartz crystals, are examples of such problems  \citep{mardia2009directional}. 
Therefore, the extension of Rao's spacing test to the sphere is an important challenge for future work.


\begin{appendices}

\section{Proof of Theorem~\ref{col}}\label{secB}
To compute the moments, we introduce the Pochhammer symbol and the following Lemmas related to the  Stirling numbers of the second kind. 

First the Pochhammer symbol is defined as follows: 
	\begin{align}
		(n)_{\overline{r}}&=n(n+1)\cdots(n+r-1), \notag\\
		(n)_{\underline{r}}&=n(n-1)\cdots(n-r+1), \notag
	\end{align}
	where $(n)_{\overline{0}}=(n)_{\underline{0}}=1$. 

Next, we introduce a lemma that expresses the relationship between polynomials and Stirling numbers. 
\begin{lemma}\citep{aigner2007course}\label{prop1}
\[
x^n=\sum_{k=0}^n (x)_{\underline{k}} S_2(n, k). 
\]
Note that we can stop the summation at n, since $S_2(n, k) = 0$
for $k=n+1, n+2, \dots $. 
\end{lemma}

\begin{lemma}\citep{aigner2007course}\label{prop2}
The recurrence relation for the Stirling numbers of the second kind is given as follows: 
	\begin{align}
		S_2(n+1,j)=S_2(n,j-1)+jS_2(n,j). \notag
	\end{align}
\end{lemma}
Here, the Stirling number of the second kind is denoted by $S_2(n,k)$. 
The $S_2(n, k)$, which represents the number of $k$-partitions of an $n$-set. 
By definition, $S_2(0, 0) = 1$, $S_2(0, k)=0 (k=1, 2, \dots)$. 

Here, we provide the following expansions involving N\"{o}rlund polynomials to derive the moments of $U_n$.
\begin{lemma}\label{lemma1}
The $r$th moment of $U_n$ is given by
\begin{eqnarray*}
E(U_n^r)=\frac{(2\pi)^r}{n^{n+r-1}}\sum_{j=1}^{n-1}(n)_{\underline{j}}\binom{n-1}{j}B_{n+r-j-1}^{(-j)}, 
\end{eqnarray*}
where $B_n^{(z)}$ is the N\"{o}rlund polynomial, defined as the coefficient in the expansion of the following expression \citep{carlitz1960note}. 
\begin{eqnarray*}
\left(\frac{x}{e^x-1}\right)^z=\sum_{n=0}^\infty B_n^{(z)}\frac{x^n}{n!}. 
\end{eqnarray*}
\end{lemma}

In the equation of Lemma~\ref{lemma1}, the N\"{o}rlund polynomial $B_n^{(z)}$ can be expressed in terms of the Stirling numbers of the second kind $S_2$ as follows \citep{quaintance2015combinatorial}: 
\begin{align*}
S_2(n+z, z)=
\left(
\begin{array}{c}
n+z\\
n
\end{array}
\right)
B_{n}^{(-z)}. 
\end{align*}
Then 
\begin{eqnarray*}
E(U_n^r)&=&\frac{(2\pi)^r}{n^{n+r-1}}\sum_{j=1}^{n-1}(n)_{\underline{j}}\binom{n-1}{j}B_{n+r-j-1}^{(-j)}\\
&=&\frac{(2\pi)^r}{n^{n+r-1}}\sum_{j=1}^{n-1}(n)_{\underline{j}}\binom{n-1}{j}
\frac{S_2(n+r-1, j)}{\binom{n+r-1}{j}}\\
&=&\frac{(2\pi)^r}{n^{n+r-1}}\sum_{j=1}^{n-1}
\frac{(n)_{\underline{j}}(n-1)_{\underline{j}}}{(n+r-1)_{\underline{j}}}S_2(n+r-1, j)\\
&{=}&{\frac{(2\pi)^r}{n^{n+r-1}}I_{n,r,0} },
\end{eqnarray*}

where 
\begin{align*}
I_{n,r,u}=\sum_{j=1}^{n-1}A_{n,j,r}(n)_{\underline{j}}S_2(n-1+r+u,j), \notag
\end{align*}
with 	
\begin{align*}
A_{n,j,r}&=\frac{(n-1)!}{(n-1-j)!}\frac{(n-1+r-j)!}{(n-1+r)!}.\notag\\
\end{align*}

To compute the moments recursively, we introduce the following theorem. 
Using recurrence relations and properties of Stirling numbers, we establish the following theorem: 

\begin{theorem}
\begin{align}
		I_{n,r,u}=\frac{1}{(n)_{\overline{r}}}\sum_{j=1}^ra_j^{(r)}(n)_{\underline{j}}(n-j)^{n+r-1+u}\label{main}
	\end{align}
where $a_j^{(r)}$ satisfies: 
\begin{align}
	a_{r}^{(r)}&=1, \notag\\
	a^{(r+1)}_1&=(r+1)a_1^{(r)},\notag\\
	a^{(r+1)}_j&=(r+j)a_j^{(r)}+a_{j-1}^{(r)}\qquad(j=2,\ldots,r).\notag
\end{align}

\end{theorem}

\begin{proof}
The proof is by induction on $r$.

For $r=1$, the following result follows from Lemma~\ref{prop1}, which characterizes the relationship between the Stirling numbers of the second kind and certain polynomials. 
\begin{align*}
I_{n,1,u}=\frac{1}{(n)_{\overline{1}}}\sum_{j=1}^1(n)_{\underline{j}}(n-j)^{n+1-1+u}. 
\end{align*}

Assuming~\eqref{main} holds for $r$, we will prove it for $r+1$.
Using Lemma~\ref{prop2} and the definitions of $I_{n,r,u}$ and $A_{n,j,r}$, we derive: 
\begin{align*}
I_{n,r+1,u}&=I_{n,r,u+1}-D_1+D_2, 
\end{align*}
where
\begin{align}
	D_1&=\frac{1}{n+r}\sum_{j=1}^{n-1}A_{n,j,r}(n)_{\underline{j}}S_2(n+r+u+1,j) ,\notag\\
	D_2&=\frac{1}{n+r}\sum_{j=1}^{n-1}A_{n,j,r}(n)_{\underline{j}}S_2(n+r+u,j-1).\notag
\end{align}
Moreover, we have

\begin{align}
	D_1=\frac{1}{n+r}I_{n,r,u+2}\label{sub3}
\end{align}

and

\begin{align}
	D_2&=\frac{n}{n+r}\sum_{j=1}^{n-1}A_{n,j,r}(n-1)_{\underline{j-1}}S_2(n+r+u,j-1)\notag\\
	&=\frac{n}{n+r}\frac{n-1}{n+r-1}I_{n-1,r,u+2}. \label{sub4}
\end{align}
The index of summation is inappropriate for $n=2$; however, this is not essential because $I_{1,r,u+2}=0$.
By summarizing the above results and performing the necessary calculations, we obtain the following.

\begin{align}
I_{n,r+1,u}&=I_{n,r,u+1}-\frac{1}{n+r}I_{n,r,u+2}+\frac{n}{n+r}\frac{n-1}{n+r-1}I_{n-1,r,u+2}\notag\\
&=\frac{1}{(n)_{\overline{r+1}}}\sum_{j=1}^{r+1}a_j^{(r+1)}(n)_{\underline{j}}(n-j)^{n+r+u},\label{final}
\end{align}

where

\begin{align*}
	a_{r+1}^{(r+1)}&=1, \\
	a^{(r+1)}_1&=(r+1)a_1^{(r)},\\
	a^{(r+1)}_j&=(r+j)a_j^{(r)}+a_{j-1}^{(r)}\qquad(j=2,\ldots,r).
\end{align*}

From the above, by mathematical induction, Equation~\eqref{main} holds for any $r$.
\end{proof}

\section{Gram-Charlier expansion}
In Theorem~\ref{col}, we recursively calculated the raw moments $\mu'_{n,r}\coloneqq E(U_n^r)$ for all $n=2, 3, \dots$ and $r=1, 2, \dots$.
Moreover, the raw cumulants $\kappa'_{n,r}$ are obtained from the raw moments  \citep{Smith01051995}:
\begin{align}
\kappa'_{n,r} = \mu'_{n,r} - \sum_{k=1}^{r-1} \binom{r-1}{k-1} \kappa'_{n,k} \mu'_{n,r-k}, \notag 
\end{align}
for all $n=2, 3, \dots$ and $r=1, 2, \dots$.
In this subsection, we calculate the distribution function of $U_n$ based on the Gram-Charlier expansion using cumulants $\kappa'_{n,r}$.
Let $\phi$ be the probability density function of the standard normal distribution and $D$ be the differential operator.
Then, the probability function $f_n(u)$ can be expanded as follows using cumulants $\kappa'_{n,r}$ \citep{stuart2010kendall}:
\begin{align*}
	f_n(u)=\exp\left(\sum_{j=3}^{\infty}\kappa'_{n,j}\frac{(-D)^j}{j!}\right)\phi(u). 
\end{align*}
Moreover, the right-hand side can be expressed in terms of the $j$th exponential complete Bell polynomials $B_j $.
\begin{align*}
	f_n(u)=\sum_{j=3}^\infty B_j(0,0,\kappa'_{n,3},\ldots,\kappa'_{n,j})\frac{(-D)^j}{j!}\phi(u), 
\end{align*}
where $B_j$'s denote the Bell polynomials defined by Definitions 1 and 2. 
  \begin{definition}[exponential partial Bell polynomials]
\begin{align}
B_{j,k}(x_1,\ldots,x_{j-k+1})&=\sum\frac{j!}{c_1!\cdots c_{j-k+1}!(1!)^{c_1}\cdots((j-k+1)!)^{j-k+1}}x_1^{c_1}\cdots x_{j-k+1}^{c_{j-k+1}} \notag\\
&=j!\sum\prod_{i=1}^{j-k+1}\frac{x_i^{c_i}}{(i!)^{c_i}c_i!},\notag
\end{align}
where the summation takes place over all integers $c_1,\ldots,c_{j-k+1}\geq0$(partitions of $j$ into $k$ summands), such that:
\begin{align}
c_1+\cdots+c_{j-k+1}=k,\notag\\
c_1+\cdots+(j-k+1)c_{j-k+1}=j.\notag
\end{align}
\end{definition}
\begin{definition}[exponential complete Bell polynomials]
\begin{align*}
B_j=\sum_{k=1}^jB_{j, k}\ . 
\end{align*}
\end{definition}
We use the exponential complete Bell polynomials $B_n$ \citep{comtet2012advanced}. 

Let $H_j(j=0,1,\ldots)$ be the Helmite polynomials
\begin{align*}
	H_0(x)&=1, \\
	H_{(j+1)}(x)&=xH_j(x)-\frac{d}{dx}H_j(x)\qquad(j=0,1,\ldots). 
\end{align*}
Then, the probability density function $f_n(u)$ can be rewritten as
\begin{align*}
	f_n(u)=\sum_{j=3}^\infty B_j(0,0,\kappa'_{n,3},\ldots,\kappa'_{n,j})\frac{H_j(u)}{j!}\phi(u), 
\end{align*}
since
\begin{eqnarray*}
D^j\phi(u)=(-1)^jH_j(u)\phi(u).
\end{eqnarray*}
Integrating both sides, we obtain the following:
\begin{align*}
	F_n(u)=\Phi(u)-\phi(u)\sum_{j=3}^\infty B_j(0,0,\kappa'_{n,3},\ldots,\kappa'_{n, j})\frac{H_{j-1}(u)}{j!} , 
\end{align*}
where $\Phi$ denotes the distribution function of the standard normal distribution and $F$ denotes the distribution function of $U_n$.
The probability $ \Pr(U_n > U^*_n) $ can be approximated by the partial sum up to the $d$th term of this infinite series. 
\begin{align}
	\Pr(U_n>U^*_n)\simeq 1-\Phi(u)+\phi(u)\sum_{j=3}^d B_j(0,0,\kappa'_{n,3},\ldots,\kappa'_{n,j})\frac{H_{j-1}(u)}{j!}.\notag 
\end{align}

\section{Higher-Order Moments of \texorpdfstring{$U_n$}{Un}}\label{sec11} 
Using the recurrence relations derived in Section~\ref{subsubsec1}, we compute the coefficients for the $r$th moments of the Rao's spacing test statistic up to order $10$. 
Table~\ref{table_a} displays the coefficient matrix $a_j^{(r)}$, where each row corresponds to a moment order $r$ and the columns represent the coefficients $a_j^{(r)}$. 
\begin{table}[htbp]
\centering
\caption{Coefficients of the $r$th moment. $\left(a_j^{(r)}\right)$}
\label{table_a}
\small
\begin{tabular}{r|rrrrrrrrrr}
 \hline
$r$ & $a_1^{(r)}$ & $a_2^{(r)}$ & $a_3^{(r)}$ & $a_4^{(r)}$ & $a_5^{(r)}$ & $a_6^{(r)}$ & $a_7^{(r)}$ & $a_8^{(r)}$ & $a_9^{(r)}$ & $a_{10}^{(r)}$ \\
 \hline
1 & 1 & 0 & 0 & 0 & 0 & 0 & 0 & 0 & 0 & 0 \\
2 & 2 & 1 & 0 & 0 & 0 & 0 & 0 & 0 & 0 & 0 \\
3 & 6 & 6 & 1 & 0 & 0 & 0 & 0 & 0 & 0 & 0 \\
4 & 24 & 36 & 12 & 1 & 0 & 0 & 0 & 0 & 0 & 0 \\
5 & 120 & 240 & 120 & 20 & 1 & 0 & 0 & 0 & 0 & 0 \\
6 & 720 & 1800 & 1200 & 300 & 30 & 1 & 0 & 0 & 0 & 0 \\
7 & 5040 & 15120 & 12600 & 4200 & 630 & 42 & 1 & 0 & 0 & 0 \\
8 & 40320 & 141120 & 141120 & 58800 & 11760 & 1176 & 56 & 1 & 0 & 0 \\
9 & 362880 & 1451520 & 1693440 & 846720 & 211680 & 28224 & 2016 & 72 & 1 & 0 \\
10 & 3628800 & 16329600 & 21772800 & 12700800 & 3810240 & 635040 & 60480 & 3240 & 90 & 1 \\
\hline
\end{tabular}
\end{table}

These coefficients can be used to compute the $r$th raw moment of $U_n$ using the general formula given in Theorem~\ref{col}.


\end{appendices}
\break

\bibliographystyle{apalike}
\bibliography{bibtex_k}

\end{document}